\documentclass[runningheads]{llncs}
\usepackage{amssymb}
\usepackage{amsmath}
\usepackage{color}
\usepackage{graphicx}
\usepackage [ruled,linesnumbered]{algorithm2e}

\newcommand{\Z }{\mathbb{Z}}
\newcommand{\R }{\mathbb{R}}

\begin{document}

\title{ Approximation Algorithms for Generalized MST and TSP in Grid Clusters }

\author{Binay Bhattacharya \inst{1}\and Ante \'Custi\'c \inst{1}
\and Akbar Rafiey \inst{1} \and Arash Rafiey \inst{1,2} \and
Vladyslav Sokol\inst{1} }

\authorrunning{B. Bhattacharya\and A. \'Custi\'c \and A. Rafiey \and A. Rafiey \and
V. Sokol}

\institute{Simon Fraser University, Burnaby, Canada,
\email{binay,acustic,arafiey,arashr,vsokol@sfu.ca}\thanks{supported by
NSERC Canada}
\and
Indiana State University, Terre Haute, IN, USA,\email{arash.rafiey@indstate.edu}
}

\maketitle

\begin{abstract}

We consider a special case of the generalized minimum spanning tree problem (GMST)
and the generalized travelling salesman problem (GTSP) where
we are given a set of points inside the integer grid (in Euclidean plane) where each gride cell is $1 \times 1$. 
In the MST version of the problem, the goal is to find a minimum tree that contains exactly 
one point from each non-empty grid cell (cluster).
Similarly, in the TSP version of the problem, the goal is to find a minimum weight cycle
containing one point from each non-empty grid cell. We 
give a $(1+4\sqrt{2}+\epsilon)$ and $(1.5+8\sqrt{2}+\epsilon)$-approximation algorithm 
for these two problems in the described setting, respectively.

Our motivation is based on the problem posed in \cite{FGS06} for a constant 
approximation algorithm. The authors designed a PTAS for the more 
special case of the GMST where non-empty cells are connected end dense enough.
However, their algorithm heavily relies on this connectivity restriction and is unpractical. 
Our results develop the topic further.

\end{abstract}


\section{Introduction}

The \emph{generalized minimum spanning tree problem} (GMST) is a generalization of the well known \emph{minimum spanning tree problem} (MST). An instance of the GMST is given by an undirected graph $G=(V,E)$ where the vertex set is partitioned into $k$ \emph{clusters} $V_i$, $i=1,\ldots,k$, and a weight $w(e)\in \R^+$ is assigned to every edge $e\in E$. The goal is to find a tree 
with minimum weight containing one vertex from each cluster.

The GMST occurs in
telecommunications network planning, where a network of node
clusters need to be connected via a tree architecture using exactly
one node per cluster \cite{GRS05}. More precisely, local subnetworks must
 be interconnected by a global network
containing a gateway from each subnetwork. For this inter-networking, a point
has to be chosen in each local network as a hub and the hub point must be
connected via transmission links such as optical fiber, see \cite{MLT95}.
Furthermore, the GMST has some applications in design of backbones
 in large communication networks, energy distribution, and agricultural
irrigation \cite{JC10}.

The GMST was first introduced by Myung, Lee and Tcha  in 1995 \cite{MLT95}. Although MST  is polynomially
solvable \cite{GJ79}, it was shown in \cite{MLT95} that the
GMST is strongly NP-hard and there is no constant
factor approximation algorithm, unless P=NP. However, several heuristic algorithms have been suggested for the GMST, see \cite{GRS05,JC10,OCL08,P02}. Furthermore, Pop, Still and Kern \cite{PSK05} used an LP-relaxation to develop 
a $2\rho-$approximation algorithm for the GMST where the
size of every cluster is bounded by $\rho$.

In \cite{FGS06}, Feremans, Grigoriev and Sitters consider the 
\textit{geometric generalized minimum spanning tree problem in grid clusters}, GGMST for short. In this special case of the GMST,  a complete graph $G=(V,E)$ is given where the set of vertices $V$ correspond to a set of points in the planar integer grid. Every non-empty $1\times 1$ cell of the grid forms a cluster. The weight of the edge between two vertices is given by their Euclidean distance. Fig.~\ref{fig:instance} depicts one instance of the GGMST.
\begin{figure}[h]
	\vspace{-10pt}
	\centering
	\includegraphics[scale=1]{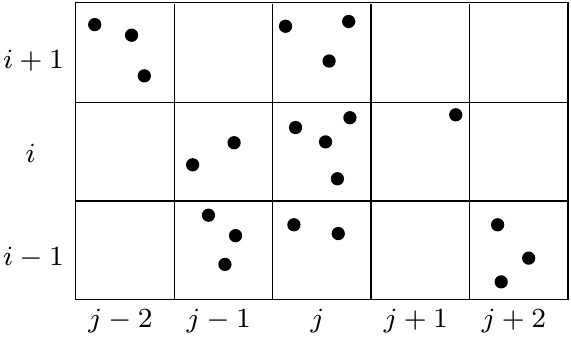}
	\vspace{-5pt}
	\caption{An GGMST instance with $n=21$ points and $N+1=8$ non-empty cells, which are connected and fit into a $3\times 5$ sub-grid}
	\label{fig:instance}
	\vspace{-10pt}
\end{figure}
We say that two grid cells are connected if they share a side or a corner. Furthermore, we say that a set of grid cells is connected if they form one connected component. 
The authors in \cite{FGS06} show that the GGMST
is strongly NP-hard, even if we restrict to instances in which non-empty grid cells are connected and each grid cell contains at most two points.
Furthermore, they designed a dynamic programming algorithm that solves in $\mathcal{O}(l\rho^{6k}2^{34k^2}k^2)$ time the GGMST for which the set of non-empty grid cells is connected and fits into $k\times l$ sub-grid. (Note that the algorithm is polynomial if $k$ is bounded.) 
Moreover, the authors used this algorithm to develop a polynomial time
approximation scheme (PTAS) for the GGMST for which non-empty cells are connected and the number of non-empty cells is superlinear in $k$ and $l$.
The GGMST instances are often used to test heuristics for the GMST which, in light of the results in \cite{FGS06}, is not adequate.
The objective of this paper is to develop this topic further and to design a
simple approximation algorithms for the GGMST and of its variants without restricting only to connected and dense instances.

Analogously as the GMST and the GGMST, the \emph{generalized travelling salesman problem} (GTSP) and the \emph{geometric generalized travelling salesman problem in grid clusters} (GGTSP) can be defined. The GTSP was introduced by Henry-Labordere \cite{L69} and is also known in the literature as \textit{set TSP}, \textit{group TSP} or \textit{One-of-a-Set TSP}. This problem has many applications, 
including airplane routing, computer file sequencing, and postal delivery, see \cite{B83,L92,LAS96}.
Elbassioni, Fishkin, Mustafa and Sitters \cite{EFMS05} considered the GTSP in which non-empty clusters (i.e.\@ regions) are disjoint $\alpha$-fat objects with possibly varying size. In this setting they obtained a $(9.1\alpha+1)$-approximation algorithm. They also give the first $\mathcal{O}(1)$-approximation algorithm for the problem with intersecting clusters (regions). Note that in the GGTSP, fatness of each cluster is $4$ (each cluster is a square).

As a special case of the GTSP we can look at each geometric region as an infinite set of points. This problem, called the \textit{TSP with neighbourhood}, was introduced by Arkin and Hassin \cite{AH94}. In the same paper they present constant factor approximation algorithm for two cases in which the regions are translates of disjoint convex polygons, and for disjoint unit disks. For the general problem Mata and Mitchell \cite{MM95} and later on Gudmundsson and Levcopoulos \cite{GL00}, gave an $\mathcal{O}(\log n)$-approximation algorithm.
For intersecting unit disks an $\mathcal{O}(1)$-approximation algorithm is given in \cite{DM01}. Safra and Schwartz \cite{SS03} show that it is NP-hard to approximate the TSP with neighbourhood within $(2-\epsilon)$. In this context, it is natural to consider the GTSP in which points are sitting inside geometric objects such as the integer grid.

{\bf Notation.\@} We will usually refer to vertices as points. Throughout  this paper, the number of points ($|V|$) will be denoted by $n$. Furthermore, $N$ denotes the number of edges in every feasible solution (tree) of the GGMST, i.e.\@ $N$ is the number of non-empty cells minus $1$. 
The edge between two points $u$ and $v$ will be denoted by $e_{u,v}$. We naturally extend the notation for the weight to sets of edges and graphs, i.g.\@ the weight of a tree $T$ is denoted by $w(T)=\sum_{e\in T}w(e)$, where $e\in T$ means that $e$ is an edge of $T$. 
We assume that every point is in just one cell, i.e.\@ points on the cell borders are assigned to only one neighbouring cell by any rule. An optimal solution of the GGMST will be denoted by $T_{opt}$ throughout this paper.

{\bf Our results and organization of the paper.\@}
The main result of this paper is a $(1+4\sqrt{2}+\epsilon)$-approximation algorithm for the GGMST. We do not assume any restrictions on connectivity, density or cardinality of non-empty cells. The algorithm is presented and analyzed in Section~\ref{sec:alg}. A lower bound for the weight of an optimal solution in terms of $N$ is used to prove the approximation ratio of the algorithm. Section~\ref{sec:lb} is devoted to proving this lower bound. Lastly, in Section~\ref{sec:tsp} we use our GGMST algorithm to develop an approximation algorithms for the GGTSP.


\section{The GGMST Approximation Algorithm}\label{sec:alg}

In this section we present a ($1+4\sqrt{2}+\epsilon$)-approximation algorithm (Algorithm~\ref{algEpsilon}) for the GGMST. Main part of the algorithm is Algorithm~\ref{alg5} which we describe next.

\begin{algorithm}[h]
$T\leftarrow$ solution of the MST problem on non-empty cells (where the distance\\ \Indp between a pair of cells is the length of the shortest edge between them)\;\Indm 
$G\leftarrow$ the graph consisting of the set of edges (and points) that correspond to the\\ \Indp edges in $T$\;\Indm
\For{\emph{all cells} $C$ \emph{that contain more than one point from} $G$}{
	$C_G\leftarrow$ the set of points from $G$ that are in $C$\;
	$p\leftarrow$ point from $C$ that is a median for $C_G$\;
	Replace $C_G$ by $p$, i.e.\@ reconnect to $p$ all edges of $G$ that enter $C$\;
}
\Return $G$\;
 \caption{$\left(1+4\sqrt{2}+\frac{2\sqrt{2}}{w(T_{opt})}\right)$-approximation alg.\@ for the GGMST}
\label{alg5}
\end{algorithm}

Algorithm~\ref{alg5} is divided into two parts; in the first part we solve an MST instance defined as follows: non-empty cells play the role of vertices, and the weight of the edge between two cells $C_1,C_2$ is the smallest weight edge $e_{p_1,p_2}$ where $p_1\in C_1$ and $p_2\in C_2$. Let $T$ be an optimal tree of such MST instance, and let graph $G$ be the set of edges (with its endpoints) of the original GGMST instance that correspond to the edges of $T$. 
Note that $G$ has $N$ edges and spans all non-empty cells but it can have multiple points in some cells. In the second part of the Algorithm~\ref{alg5} (i.e.\@ the {\bf for} loop), we modify $G$ to obtain the GGMST feasibility, by iteratively replacing multiple cell points by a single point $p$. We choose point $p$ to be the one that has the minimum sum of distances to other points of $G$ that are in the corresponding cell.

Next we present an upper bound for solutions obtained by Algorithm~\ref{alg5} in terms of the number of edges $N$.

\begin{theorem}\label{thm3}
	Algorithm~\ref{alg5} produces a feasible solution $T_A$ with $w(T_A)\leq w(T_{opt})+\sqrt{2}N-\sqrt{2}$, where $N$ is the number of edges of $T_A$.
\end{theorem}
\begin{proof}
	Denote by $G_0$ the non-feasible graph obtained in the first part of the algorithm, i.e.\@ the first version of graph $G$. Then the weight of the solution $T_A$ obtained by the algorithm is equal to $w(G_0)+ext$, where $ext$ is the amount by which we increase (extend) the weight of $G_0$ in the second part of the algorithm. Note that $w(G_0)\leq w(T_{opt})$, as $G_0$ is an optimal solution of the problem for which $T_{opt}$ is a feasible solution (find a minimum weight set of edges that spans all non-empty cells, with all GGMST edges being allowed). In the rest of the proof we will bound the value of $ext$.

In every run of the \textbf{for} loop we replace the set of points $C_G$ with $p$. In doing so, every edge $e_{q,c}$, $c\in C_G$ from $G$, is replaced by $e_{q,p}$. From the triangle inequality we get that $w(e_{q,p})-w(e_{q,c})\leq w(e_{c,p})$. Hence, the increase (extension) of the weight of $G$ in every run of the \textbf{for} loop is less or equal than $\sum_{c\in C_G}w(e_{c,p})$. Instead of bounding such absolute values, we will bound its average per edge adjacent to the corresponding cell.  
More precisely, we will calculate an average extension per \emph{half-edge} assigned to the corresponding cell. Namely, every edge will be extended at most two times, once on each endpoint, so we can look at each extension as an extension of a {half-edge}. Furthermore, note that edges that contain leafs will be extended only on one side. We will use this fact to assign half-edges that contain leafs to other cells to lower their average half-edge extension. 
To every cell $C$, we will assign $|C_G|-2$ leaf half-edges. Intuitively, we can do this because every node $v$ of a tree \emph{generates} $\deg(v)-2$ leafs. Formally, it follows from the following well known equality:
\begin{equation}\label{eq3}
|V_1|=2+\sum_{i\geq2}|V_i|(i-2),
\end{equation}
where $V_i=\{v\in V\colon \deg(v)=i\}$, and $V$ is the set of vertices of a graph.

Then for a cell $C$ the average extension per assigned half-edges is bounded above by 
\begin{equation}\label{eq4}
	\frac{\sum_{c\in C_G}w(e_{c,p})}{|C_G|+(|C_G|-2)}.
\end{equation}
Note that the maximum distance between two cell points is $\sqrt{2}$. Since points from $C_G$ are candidates for $p$, it follows that $\sum_{c\in C_G}w(e_{c,p})\leq \sqrt{2}(|C_G|-1)$. Hence, \eqref{eq4} is bounded above by
\[
\frac{\sqrt{2}(|C_G|-1)}{2|C_G|-2}=\frac{\sqrt{2}}{2}.
\]
Hence, in average, every half-edge (except 2 leaf half-edges, see \eqref{eq3}) is extended by at most $\sqrt{2}/2$. Note that this average bound is a constant, i.e.\@ does not depend on $C$. Now $ext$ can be bounded by 
\begin{equation}
ext\leq \frac{\sqrt{2}}{2}(2N-2)=\sqrt{2}N-\sqrt{2}.
\end{equation}
Finally, we can bound the solution $T_A$ of the algorithm by
\[
w(T_A)\leq w(G_0)+ext\leq w(T_{opt})+\sqrt{2}N-\sqrt{2}.
\]
\qed\end{proof}

The following theorem gives a lower bound for the optimal solution in terms of the number of edges $N$. Section~\ref{sec:lb} is dedicated to proving the theorem.

\begin{theorem}\label{thm:lb}
If $T_{opt}$ is an optimal solution of the GGMST on $N+1$ non-empty cells, then $N\leq  4w(T_{opt})+3$.
\end{theorem}

Now from Theorem~\ref{thm3} and Theorem~\ref{thm:lb} the following approximation bound for Algorithm~\ref{alg5} follows.
\begin{corollary}\label{cor:alg1}
	Algorithm~\ref{alg5} produces a feasible solution $T_A$ of the GGMST such that  $w(T_A)\leq  (1+4\sqrt{2})w(T_{opt})+2\sqrt{2}$.
\end{corollary}

Note that, due to the constant $2\sqrt{2}$, Corollary~\ref{cor:alg1} does not gives us a constant approximation ratio for Algorithm~\ref{alg5}. Namely, the approximation ratio that we get is equal to $1+4\sqrt{2}+\frac{2\sqrt{2}}{w(T_{opt})}$. 
Next we focus on improving Algorithm~\ref{alg5} so that $\frac{2\sqrt{2}}{w(T_{opt})}$ is replaced by arbitrary small $\epsilon>0$. Note that the optimal solution weight does not necessarily increase with the increase of the number of points $n$, namely all points can be in the same cells. Hence we cannot use the standard approach. However, the following two facts will do the trick. First, note that the  weight of the GGMST optimal solution increases as the number of non-empty cells increases. Second, given a spanning tree structure of non-empty cells $T$, we can in polynomial time find the minimum weight GGMST feasible solution $T'$ with the same tree structure as $T$ (i.e.\@ there is an edge in $T'$ between two cells if and only if these two cells are adjacent in $T$).  Next we design a dynamic programming algorithm to find $T'$ (see  Algorithm~\ref{algDP}).

Given an GGMST instance, let $T$ be a spanning tree of the complete graph where the set of vertices correspond to the set of non-empty cells. Denote by $X_i$ the set of points inside cell $C_i$. We observe $T$ as a rooted tree with $C_r$ as its root. If $C_i$ is a leaf of $T$ then the weight $W(z)$ of each point $z$ in set $X_i$ is set to zero. If $C_i$ is not a leaf then $T$ has some children $C_{i_1},\ldots,C_{i_k}$ and the weight for points inside sets $X_{i_1},\ldots,X_{i_k}$ has already been computed. Then for each point $p$ in cell $C_i$ (set $X_i$) we compute:
\[
W(p)=\sum_{j=1}^k\min_{q\in X_{i_j}}\{W(q)+w(e_{p,q})\}
\]
Algorithm~\ref{algDP} computes $W(p)$ for all $p\in C_r$. Note that it is easy to adapt Algorithm~\ref{algDP} to store selected points at each step.

\begin{algorithm}[h]
\KwData{A spanning tree $T$ of non-empty cells}
\KwResult{An optimal weight of the GGMST tree with the same structure as $T$}
Choose an arbitrary cell $C_r$ as the root of $T$\;
\For{\emph{each leaf} $C_i$ \emph{of} $T$}{
	\For{\emph{each} $p\in X_i$ }{
		$W(p)=0$\;
	}
}
$CurrentLevel$ = height of $T$\;
\While{CurrentLevel $\geq$ \emph{root level}}{
	\For{\emph{each node} $C_i$ \emph{of} CurrentLevel}{
		Let $C_{i_1},\ldots,C_{i_k}$ be children of $C_i$ in $T$\;
		\For{\emph{each} $p\in X_i$}{
			$W(p)=\sum_{j=1}^k\min_{q\in X_{i_j}}\{W(q)+w(e_{p,q})\}$\;
		}
	}
	$CurrentLevel = CurrentLevel-1$\;
}
\Return $\min_{p\in X_r}W(p)$\;
 \caption{Optimal GGMST solution for a given spanning tree of cells}
\label{algDP}
\end{algorithm}

Now we have all ingredients  to design a $(1+4\sqrt{2}+\epsilon)$-approximation algorithm, see Algorithm~\ref{algEpsilon}. Note that $1+4\sqrt{2}$ is approximately equal to $6.66$.

\begin{algorithm}[h]
\eIf{$N\leq 15$ \bf{or} $N\leq 10\sqrt{2}/\epsilon$}{
	Output minimum weight solution obtained by Algorithm~\ref{algDP} on all spanning trees of non-empty cells\;
}{
	Run Algorithm~\ref{alg5}\;
}
 \caption{$(1+4\sqrt{2}+\epsilon)$-approximation algorithm for the GGMST}
\label{algEpsilon}
\end{algorithm}

\begin{theorem}\label{thmEpsilon}
	For any $\epsilon>0$, Algorithm~\ref{algEpsilon} is a $(1+4\sqrt{2}+\epsilon)$-approximation algorithm for the GGMST.
\end{theorem}
\begin{proof}
	If $N\leq 15$ or $N\leq 10\sqrt{2}/\epsilon$, then we can enumerate all spanning trees on $N+1$ non-empty cells, and apply Algorithm~\ref{algDP} on each of them. That will give us an optimal solution in polynomial time.

Assume $N>15$ and $N> 10\sqrt{2}/\epsilon$. By Corollary~\ref{cor:alg1} it follows that Algorithm~\ref{alg5} will produce a solution $T_A$ such that
\begin{equation}\label{eq2}
w(T_A)\leq \left(1+4\sqrt{2}\right)w(T_{opt})+2\sqrt{2}.
\end{equation}
From Theorem~\ref{thm:lb} and $N>15$ it follows that $1\leq 5w(T_{opt})/N$. Applying that on the rightmost element of inequality \eqref{eq2} we get
\begin{align*}
w(T_A)&\leq \left(1+4\sqrt{2}\right)w(T_{opt})+\frac{10\sqrt{2}}{N}w(T_{opt}),\\
 &\leq \left(1+4\sqrt{2}+\frac{10\sqrt{2}}{N}\right)w(T_{opt}).
\end{align*}
Now from $N> 10\sqrt{2}/\epsilon$ it follows that
\[
w(T_A)\leq \left(1+4\sqrt{2}+\epsilon\right)w(T_{opt}),
\]
which proves the theorem.
\qed\end{proof}


\section{The Lower Bound Proof}\label{sec:lb}

This section is entirely devoted to proving Theorem~\ref{thm:lb} which gives us a lower bound on the weight of an optimal solution. The lower bound is expressed in terms of the number of edges $N$. 

Throughout this section we identify $1\times 1$ grid cell with its coordinates $(i,j)$, where $i,j\in\Z$ is the row and the column of the cell inside the infinite integer grid. For example, in Fig.~\ref{fig:instance}, cell $(i,j+1)$ contains one point which is near its upper right corner.

We start by proving lower bounds for trees of small size. 

\begin{lemma}\label{lm:4}
The weight of any subtree of $T_{opt}$ with four edges is at least $1$.
\end{lemma}

\begin{proof}
Consider a subtree $T'$ of $T_{opt}$ with four edges.
Let $H$ denote the set of the five cells that contain vertices of $T'$. Note that there will be two cells in $H$ with coordinates $(i,j)$ and $(i',j')$ such that $|i-i'|\geq 2$ or $|j-j'|\geq 2$.  
Hence, Euclidean distance between a vertex from the cell $(i,j)$ and a vertex from the cell $(i',j')$ is a least $1$.
This implies $w(T') \ge 1$. See Fig.~\ref{fig:2} for an example.
\begin{figure}[h]
	\vspace{-10pt}
	\centering
	\includegraphics[scale=1.1]{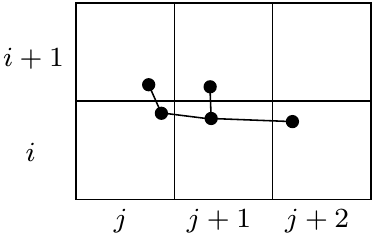}
	\vspace{-5pt}
	\caption{An example of a tree $T'$ with four edges}
	\label{fig:2}
	\vspace{-10pt}
\end{figure}
\qed\end{proof}

\begin{lemma}\label{lm:7}
The weight of any subtree of $T_{opt}$ with seven edges is at least $\frac{1}{3}(2\sqrt{6}+\sqrt{6-3\sqrt{3}})$ (which is greater than $1.93$).
\end{lemma}

\begin{proof}
Let $T'$ be a subtree of $T_{opt}$ with seven edges. If $T'$ does not fit in any $3 \times 3$ sub-grid of the original grid, then there are two 
vertices $u,v$ of $T'$ which are from cells with coordinates $(i,j)$ and $(i',j')$ such that $|i-i'| \ge
3$ or $|j-j'| \ge 3$. In that case $w(e_{u,v}) \ge 2$ and therefore
$w(T') \ge 2$. 

Next we consider the case when $T'$ fits into $3 \times 3$ sub-grid. Since $T'$ has eight vertices, at least three of them are in the corner cells of a $3\times 3$ grid. Without loss of generality we assume that these three vertices are vertex $v$ in cell $(i,j)$, vertex $u$ in cell $(i+2,j)$ and vertex $y$ in cell $(j+2,i)$. Let $P$ be a shortest path in $T'$ from $v$ to $u$ and let $Q$ be the shortest path in $T'$
from $v$ to $y$. Note that $w(e_{v,u}) \ge 1$ and $w(e_{v,y}) \ge 1$. If $P$ and $Q$ do not have a common vertex apart from $v$, then $w(T') \ge 2$.
Thus we are left with the case when $P$ and $Q$ have a common vertex other than $v$, which we denote by $x$.

First we assume that $P$ and $Q$ do not go through the point in cell $(i+1,j+1)$.
In this case, up to symmetry, one of the configurations depicted in Fig.~\ref{fig:7}~(a,b) occurs.
\begin{figure}[h]
	\vspace{-10pt}
	\centering
	\includegraphics[scale=.9]{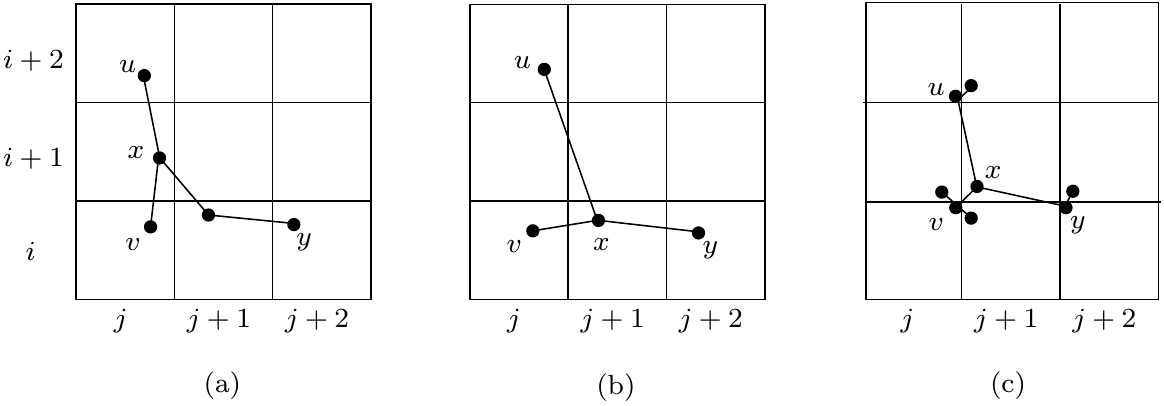}
	\vspace{-5pt}
	\caption{Layouts of $P$ and $Q$}
	\label{fig:7}
	\vspace{-10pt}
\end{figure}
However, it is clear that $w(e_{v,x})+w(e_{x,y})+w(e_{x,u}) \ge 2$ and hence $w(T') \ge 2$.

Lastly, we observe the case when vertex $x$ is in cell $(i+1,j+1)$. 
Then $w(P \cup Q)$ is at least $w(e_{x,v})+w(e_{x,u})+w(e_{x,y})$, which is minimized when $x$ is the Fermat point for the three corners of cell $(i+1,j+1)$ and $T'$ has the structure depicted in Fig.~\ref{fig:7}~(c). Therefore it can be computed that $w(T')\geq\frac{1}{3}(2\sqrt{6}+\sqrt{6-3\sqrt{3}})>1.93$.
\qed\end{proof}

\begin{lemma}\label{lm:8}
	The weight of any subtree of $T_{opt}$ with eight edges is at least $2$.
\end{lemma}

\begin{proof}
Let $T'$ be a subtree of $T_{opt}$ with eight edges. If $T'$ does not fit in any $3 \times 3$ sub-grid then by the same simple argument as in the proof of Lemma~\ref{lm:7} we get $w(T')\geq 2$. If $T'$ fits in a $3 \times 3$ sub-grid, then there is one vertex of $T'$ in any cell of such $3\times 3$ grid. More specifically, there are vertices in cells $(i,j),$ $(i+2,j),$ $(i,j+2)$ and $(i+2,j+2)$  from which easily follows that $w(T') > 2$.
\hfill$\square$\end{proof}

\begin{lemma}\label{lm:9}
	The weight of any subtree of $T_{opt}$ with nine edges is at least $1+\sqrt{3}$.
\end{lemma}

\begin{proof}
Let $T'$ be a subtree of $T_{opt}$ with nine edges. If $T'$ does not fit in any $4 \times 4$ sub-grid of the original grid, then there are two 
vertices $u,v$ of $T'$ which are in cells with coordinates $(i,j)$ and $(i',j')$ such that $|i-i'| \ge
4$ or $|j-j'| \ge 4$. In that case $w(e_{u,v}) \ge 3$ and therefore
$w(T') \ge 3>1+\sqrt{3}$. 

Next we consider the case when the smallest rectangular sub-grid that contains $T'$ is of the size $4\times 4$, and let $(i,j)$ be the bottom left corner cell  of such $4\times 4$ grid.
In that case there are four (not necessarily distinct) vertices $u,v,x,y$ of $T'$ that for some $i\leq i',i''\leq i+3$ and $j\leq j',j''\leq j+3$ lie in cells $(i',j),(i,j'),(i'',j+3),(i+3,j'')$, respectively. 
Let $P$ be the shortest path in $T'$ from $u$ to $x$ and let $Q$ be the shortest path in $T'$ from $v$ to $y$.
Let us observe the union of paths $P$ and $Q$. This union is a set of $k$ edges we denote by $e_\ell$, $\ell=1,\ldots,k$. Let us denote by $x_\ell$ and $y_\ell$ the lengths of projections of $e_\ell$ on $x$-axis and $y$-axis, respectively. Then 
\begin{equation}\label{eq:proj}
w(P\cup Q)=\sum_{\ell=1}^k \sqrt{x_\ell^2+y_\ell^2}.
\end{equation}
Since distance between projections of $u$ and $x$ on $x$-axis is at least 2 and distance between projections of $v$ and $y$ on $y$-axis is at least 2, it follows that 
$\sum_{\ell=1}^k x_\ell\geq 2$
and
$\sum_{\ell=1}^k y_\ell\geq 2$.
Hence, \eqref{eq:proj} is minimized when $k=1$ and $x_1=y_1=2$ with minimal value being $2\sqrt{2}$. Therefore we get $w(T')\geq 2\sqrt{2}> 1+\sqrt{3}$.

Lastly, we consider the case when $T'$ fits into a rectangular sub-grid $R$ of dimensions smaller than $4\times 4$. Without loss of generality we can assume that $R$ is of the size $4\times 3$, and let $(i,j)$ be the bottom left corner cell of $R$. Note that there are at least two vertices of $T'$ that are in corner cells of $R$. Without loss of generality we assume that vertex $v$ is in cell $(i,j)$. Next we distinguish remaining cases with respect to the position of the second corner point which we denote by $u$.

{\bf Case 1.} Vertex $u$ is in cell $(i,j+2)$. As there are ten vertices in $T'$, one of them must be in cell $(i+3,j')$ for some $j\leq j'\leq j+3$. Denote such vertex by $y$. By calculating the Fermat point $x$ it can be seen that weight of the Steiner tree containing $u,v$ and $y$ is at least $2+\sqrt{3}/2$ which is greater than $1+\sqrt{3}$, see Fig.~\ref{fig:9}~(a).
\begin{figure}[h]
	\vspace{-10pt}
	\centering
	\includegraphics[scale=.9]{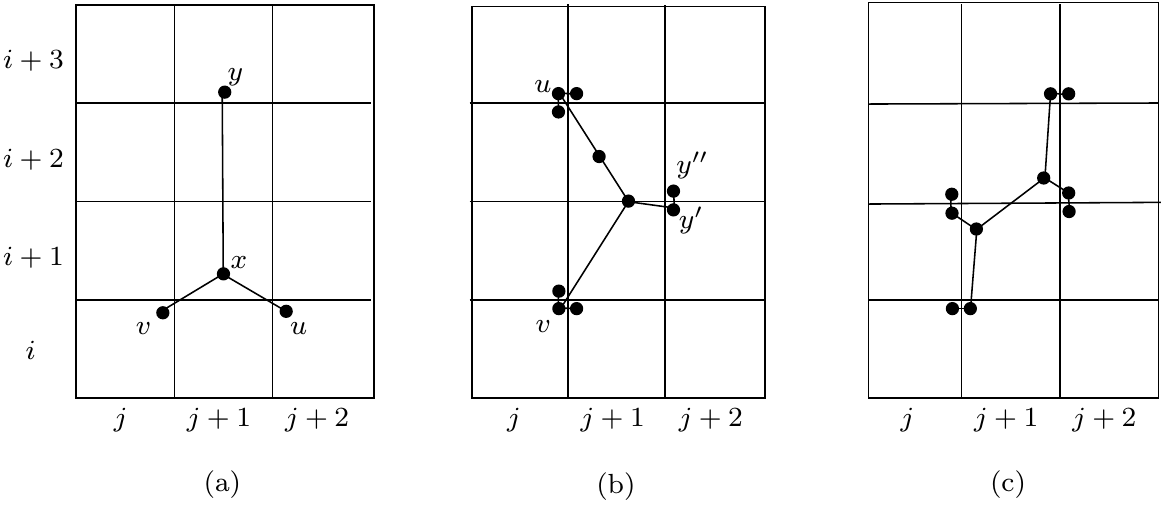}
	\vspace{-10pt}
	\caption{$T'$ configurations cases}
	\label{fig:9}
	\vspace{-10pt}
\end{figure}

{\bf Case 2.} Vertex $u$ is in cell $(i+3,j)$. We can assume that there are no vertices of $T'$ in cells $(i,j+2)$ or $(i+3,j+2)$ as then {\bf Case 1} applies. Then there must be vertices $y', y''$ in $T'$ in cells $(i+1,j+2)$ and $(i+2,j+2)$. Hence, $w(T')$ must be at least as the weight of the Steiner tree that contains right upper corner of cell $(i,j)$, right bottom corner of cell $(i+3,j)$ and left bottom corner of cell $(i+2,j+2)$. By calculating the Fermat point, one can see that such Steiner tree has weight $1+\sqrt{3}$, hence $w(T')\geq 1+\sqrt{3}$. In Fig.~\ref{fig:9}~(b) subtree $T'$ has the configuration that mimics such Steiner tree. 

{\bf Case 3.} Vertex $u$ is in cell $(i+3,j+2)$. We can assume that there are no vertices of $T'$ in cells $(i,j+2)$ or $(i+3,j)$ as then {\bf Case 1} or {\bf Case 2} apply. In this case minimal weight $T'$ mimics the Steiner tree that contains right upper corner of cell $(i,j)$, left bottom corner of cell $(i+3,j+2)$, right bottom corner of  cell $(i+2,j)$ and left upper corner of the cell $(i+1,j+2)$, see Fig.~\ref{fig:9}~(c). It is easy to calculate that the weight of such Steiner tree is $\sqrt{5+2\sqrt{3}}$ which is greater than $1+\sqrt{3}$.
\qed\end{proof}

Now we are ready to prove Theorem~\ref{thm:lb}.

\begin{proof}[of Theorem~\ref{thm:lb}]
We will proof the theorem by induction on $N$. Recall that $N$ is the number of edges in $T_{opt}$.

By Lemma~\ref{lm:4}, Lemma~\ref{lm:8} and Lemma~\ref{lm:9}, theorem holds for $N\leq 13$. Next we assume that theorem holds for all trees with number of edges strictly less than $N$.

We will perform the induction step as follows: through exhaustive case study we will show that there always exist a subtree $T'$ of $T_{opt}$ for which $w(T')$ is greater or equal to number of edges of $T'$ divided by $4$, and if we remove from $T_{opt}$ the edges of $T'$, it remains connected. In that case, by induction hypothesis the bound for $T_{opt}$ holds.

We observe $T_{opt}$ as a rooted tree, and given a vertex $v$ of $T_{opt}$, we denote by $T_v$  the maximal subtree of $T_{opt}$ rooted at $v$. 

Let $u$ be a non-leaf vertex of $T_{opt}$ with maximum number of edges in its path to the root. 

{\bf Assumption 1:} We may assume $u$ has at most two children. Namely, in the case when $u$ has four children $u_1,u_2,u_3,u_4$ let $T'$ be a subtree of $T_u$ induced by $\{u,u_1,u_2,u_3,u_4\}$. In the case when $u$ has exactly three children $u_1,u_2,u_3$ set $T'$ to be $T_v$ where $v$ is the parent of $u$. Note that in both cases $T'$ has four edges. Let $T''=T_{opt}\setminus E(T')$ where $E(T)$ denotes the set of edges of a tree $T$. Since $T''$ is a tree, by induction hypothesis it follows that $|E(T'')|=N-4 \le 4w(T'')+3$. Furthermore, by Lemma~\ref{lm:4} we have that $4 \le 4w(T')$. Hence, $N \le 4w(T'')+4w(T')+3=4w(T_{opt})+3$. 

{\bf Assumption 2:}
If $u$ has exactly two children $u_1,u_2$, we may assume that the parent of $u$ (denoted by $v$) has degree strictly greater than two. Namely, if this is not the case, we set $T'=T_v\cup\{e_{v,w}\}$ where $w$ is the parent of $v$, and we set $T''=T_{opt}\setminus E(T_w)$.
Since $T'$ has four edges and $T''$ is a tree, by induction hypothesis for $T''$ and Lemma \ref{lm:4} we obtain the bound.

{\bf Case 1:} Vertex $u$ has exactly two children $u_1,u_2$. Then by {\bf Assumption~2}
$v$ has at least two children.
By the choice of $u$, the number of edges in any path from $v$ to a leaf in $T_v$ is at most $2$.
Let $w'$ be another child of $v$. By {\bf Assumption~1} $w'$ has at most two children.
Also note that we can assume that $w'$ has at least one child. Otherwise the subtree $T'$
induced by $\{w',v,u,u_1,u_2\}$ has four edges, hence by removing the edges of $T'$ from $T$ we can
apply the induction hypothesis and obtain the bound. 

{\bf Case 1.1:} Vertex $v$ has another child $w''$. In this case using
the same arguments as above it can be shown that  $w''$ must have exactly one or two children.
Note that subtree $T'$ induced by $v,u,u_1,u_2$ together with $T_{w'},T_{w''}$ has
at least seven edges and at most nine edges. Therefore,
Lemma~\ref{lm:7}, Lemma~\ref{lm:8} or Lemma~\ref{lm:9} can be applied for each of the cases. Furthermore, for the remaining subtree $T_{opt} \setminus E(T')$ the induction hypothesis can be applied to obtain the bound.

{\bf Case 1.2:} Vertex $v$ has only two children $w',u$. Let $w$ be the parent of $v$.
We can assume that $w'$ has exactly one child, otherwise the
subtree $T'$ induced by the vertices of $T_v$ and vertex $w$ has exactly seven edges, hence we could use Lemma~\ref{lm:7}. 
If the degree of $w$ is two, then let $T'$ be the
subtree induced by $T_{w}$ together with the edge $e_{w,y}$, where $y$ is the parent of $w$.
$T'$ has seven edges and therefore, the result follows. 
Now, we may assume that $w$ has another child $v'$. Let $T_1=T_v$ and observe that $T_1$
has $5$ edges. Let $T_2=T_{v'}$. By the same argument used for $T_{v}$, we
conclude that $T_2$ has at most five edges. Let $T'=T_1 \cup T_2 \cup \{e_{w,v'},e_{w,v}\}$. 
If $T_2$ has zero, one or two edges, then $T'$ has at least seven and at most nine
edges, and hence the bound follows. If $T_2$ has four edges then by induction
hypothesis on $T_{opt}\setminus E(T_2)$ and by applying the Lemma \ref{lm:4} on $T_2$, we obtain the bound. 
It remains to consider the cases when $T_2$ has three or five edges. If $T_2$ has three edges, then
we add edge $e_{w,v'}$ to $T_2$ and now the new tree has four edges, hence we can apply the same arguments as before.  
We are left only with the case when $T_2$ has five edges. 
In this case $w(T_2) \ge 1$, according to Lemma \ref{lm:4}, and also $T_3=T_1 \cup \{e_{w,v'}, e_{v,w}\}$
has seven edges. By Lemma \ref{lm:7}, either $w(T_3)$ is at least $2$, or it has the structure depicted in Fig.~\ref{fig:7}~(c), and it is clear that every edge incident to the tree in Fig.~\ref{fig:7}~(c) is grater than, say $0.5$. Hence, in either case  $w(T') \ge 3$. Since $T'$ has twelve edges the bound is obtained by induction hypothesis on $T_{opt}\setminus E(T')$.

{\bf Case 2:} Vertex $u$ has exactly one child $u_1$.

{\bf Case 2.1} Vertex $v$ has another child $w'$. In this case $T_{w'}$ has depth at most $1$.
If $w'$ has more than one child, then from {\bf Case~1} ($w'$ instead of $u$) we are done.
If $w'$ has one child (denoted by $w_1$), then the subtree induced by $\{u_1,u,v,w',w_1\}$ has four edges and we are done.

We continue by assuming that $w'$ has no child. If $v$ has another child $w'' \notin\{u,w'\}$, then as we
argued for $w'$, we can assume that $w''$ has no child. However, in this case subtree induced by $\{u_1,u,v,w',w''\}$ has
four edges and we are done.
Therefore we can assume that $v$ has exactly two children $w'$ and $u$. Let $w$ be the parent of $v$. Then the subtree induced by $u_1,u,v,w',w$
has four edges and we are done.

{\bf Case 2.2:} Vertex $v$ has only child $u$. Let $w$ be the parent of $v$. W can assume that $v$ has a sibling node $v'$, as otherwise we can remove the four edge subtree induced by $\{u_1,u,v,w,z\}$, where $z$ is the parent of $w$. Furthermore, we can assume that $v'$ has a child $u'$, as otherwise we can remove the four edge subtree induced by $\{u_1,u,v,w,v'\}$.

{\bf Case 2.2.1:} Vertex $u'$ has no child but has a sibling $u''$. We can assume that no child of $v'$ has a child, as we can observe such case as an instance of {\bf Case 2.2.3}. Furthermore, we can assume that $u'$ and $u''$ are only children of $v'$. Otherwise, in the case when $v'$ has more than three children, there would exist a subtree of $T_{v'}$ with four edges that we could remove. Furthermore, in the case when $v'$ has exactly three children, we can remove $T'=T_{v'}\cup\{e_{w,v'}\}$. 

Hence we are left with the case when $u''$ is the only sibling of $u'$. In the case $v$ and $v'$ are only children of $w$, we can remove seven edge subtree $T'=T_w\cup\{e_{w,z}\}$, where $z$ denotes the parent of $w$. Lastly, we consider the case when there exist third child of $w$ denoted by $v''$. From the assumptions and solved cases above, we can assume that $T_{v''}$ has at most two edges, hence subtree $T'=T_v \cup T_{v'} \cup T_{v''} \cup \{e_{w,v}, e_{w,v'}, e_{w,v''}\}$ has seven, eight or nine edges, therefore we can remove it. 

{\bf Case 2.2.2:} Vertex $u'$ has no child nor sibling. In the case there exists a third child of $w$, from the assumptions and solved cases above if would follow that we can assume that it has only one child which has no child. In that case thee would exist a subtree of $T_w$ with four edges that we can remove. Hence, we can assume that $w$ has no other children besides $v$ and $v'$.  Then $T_w$ is a path with five edges. If $w(T_w)$ is grater than $5/4$, we can remove it and we are done. Otherwise it must be similar to the structure depicted in Fig.~\ref{fig:path5}, i.e.\@ with a path of approximate size $1$ alongside a border of a cell, and with remaining vertices grouped at the endpoints of such path. 
\begin{figure}[h]
	\vspace{-10pt}
	\centering
	\includegraphics[scale=1.0]{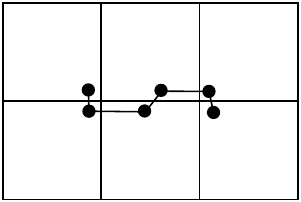}
	\caption{A short path with five edges}
	\label{fig:path5}
	\vspace{-10pt}
\end{figure}
Note that in that case, edge $e_{w,z}$ must be big enough so that $w(T_w\cup \{e_{w,z}\})$ is greater than $6/4$. Hence we can remove $T_w\cup \{e_{w,z}\}$ and by induction hypothesis obtain the bound. 

{\bf Case 2.2.3:} Vertex $u'$ has a child $u_1'$. Note that from the assumption on maximality of depth of $u$, $u_1'$ has no children. As we solved {\bf Case~2.1}, we can assume that $u_1'$ has no siblings. Furthermore, we can assume that there is no sibling of $u'$ that has a child, as in that case there would exist subtree of $T_{v'}$ with four edges that we could remove. Now in the case that $u'$ has more than one sibling, again, there would exist subtree of $T_{v'}$ with four edges that we could remove. In the case that $u'$ has exactly one sibling, subtree $T'=T_{v'}\cup\{e_{w,v'}\}$ can be removed.
We are left with the case when both $T_v$ and $T_{v'}$ are paths with two edges. In the case there is a third child of $w$, denoted by $v''$, from the solved cases above if follows that we can assume that $T_{v''}$ is also a path with two edges. In that case there is a  subtree of $T_w$ with nine edges that can be removed. In the case there is no third child of $w$, the seven edges subtree $T'=T_w\cup \{e_{w,z}\}$ (with $z$ being the parent of $w$), can be removed and the bound obtained.

We considered all the cases, therefore proving the theorem.
\qed\end{proof}


\vspace{-2mm} 

\section{Approximation of the GGTSP}\label{sec:tsp}

Our approximation algorithms for the GGMST can be used to obtain approximation algorithms for the geometric generalized travelling salesman problem on grid clusters (GGTSP) using standard methods.

We start with the approach of shortcutting a double MST, presented in Algorithm~\ref{alg:doubleMST} and analyzed next.

\vspace{-4mm}

\begin{algorithm}[h]
	\KwData{Instance $I$ of the GGTSP}
	\KwResult{Generalized travelling salesman tour}
	$T_A \leftarrow$ output of Algorithm~\ref{algEpsilon} on $I$\;
	$G_E\leftarrow$ Eulerian graph obtained by doubling all edges in $T_A$\;
 	$\mathcal{ET}\leftarrow$ an Euler tour of $G_E$\;
 	$\mathcal{C}\leftarrow$ a GGTSP tour obtained by going along $\mathcal{ET}$ and skipping repeated vertices\;
 	\Return $\mathcal{C}$;
 	\caption{$(2+8\sqrt{2}+2\epsilon)$-approximation algorithm for the GGTSP}
	\label{alg:doubleMST}
\end{algorithm}

\vspace{-4mm}

By removing one edge from a GGTSP tour, one obtains a GGMST tree, hence $w(T_A)$ is less than $(1+4\sqrt{2}+\epsilon)OPT$, where $OPT$ is the weight of an optimal solution of the GGTSP. Therefore, $w(G_E)$ is less than $2(1+4\sqrt{2}+\epsilon)OPT$. Due to triangle inequality, shorcutting the Euler tour in line 4 of the algorithm does not increase the weight. Hence, Algorithm~\ref{alg:doubleMST} is a  $(2+8\sqrt{2}+2\epsilon)$-approximation algorithm for the GGTSP. Note that $2+8\sqrt{2}$ is approximately equal to $13.31$.

Next we use the approach from the famous Christofides $\frac{3}{2}$-approximation algorithm for the metric TSP, see \cite{C76}. This approach will give us $0.5$ decrease of the 
approximation ratio. We give a sketch of the algorithm and the analysis, and leave details to the reader.

We start by running Algorithm~\ref{alg5} on the GGTSP instance. Let $T_G$ be the resulting tree. Note that $w(T_G)$ is less or equal than $(1+4\sqrt{2})OPT+2\sqrt{2}$, 
where $OPT$ is the weight of an optimal solution of the GGTSP. 
Let $S$ be a set of non-empty cells that contain a vertex of $T_G$ with an odd degree. Note that $|S|$ is even. Let $M$ be a minimum perfect 
matching among cells in $S$, where the distance between two cells $C_1,C_2\in S$ is the smallest distance between two points $p_1,p_2$ among all $p_1\in C_1$, $p_2\in C_2$. 
It is not hard to show that $w(M)\leq\frac{1}{2} OPT$. Let $M_G$ be the set of edges $e_{t_1,t_2}$ for which $t_1,t_2$ are vertices of $T_G$ and there exist an edge $e_{p_1,p_2}\in M$ 
such that $p_1$ and $t_1$ are in the same cell and  $p_2$ and $t_2$ are in the same cell. Note that $w(M_G)\leq \frac{1}{2} OPT+N\sqrt{2}$, and hence by 
Theorem~\ref{thm:lb} we get that $w(M_G)\leq \frac{1}{2} OPT+4\sqrt{2}OPT+3\sqrt{2}$. 
By merging $M_G$ and $T_G$ we obtain an Eulerian graph, and by shortcutting one of its Euler tours we obtain a GGTSP tour with weight at most 
$(\frac{3}{2}+8\sqrt{2})OPT+5\sqrt{2}$. By similar approach as in Algorithm~\ref{algEpsilon} and Theorem~\ref{thmEpsilon}, we can get rid of $5\sqrt{2}$ error, and 
obtain a $(\frac{3}{2}+8\sqrt{2}+\epsilon)$-approximation algorithm for every $\epsilon >0$. 

\vspace{-4mm} 

\section{Conclusions}

We presented a simple $(1+4\sqrt{2}+\epsilon)$-approximation algorithm for the geometric 
generalized minimum spanning tree problem on grid clusters (GGMST) and $(1.5+8\sqrt{2}+\epsilon)$-approximation 
algorithm for the geometric generalized travelling salesman problem on grid clusters (GGTSP). 

To obtain guarantied approximation ratios for our algorithms, we used the following lower bound on the optimal 
solution: Every tree with $N$ edges that contains at most one point from any $1\times 1$ grid cell is of size at least $\frac{N-3}{4}$. 
Obtaining a tight lower bound in terms of the number of edges would decrees guaranteed approximation ratios of our (and other similar) algorithms. 
Moreover, it would be an interesting result on its own.\\

{\bf Acknowledgment:} We would like to thank Geoffrey Exoo for many usefull discussions.

\end{document}